\newif\ifdraft
\newif\iflong

\longtrue

\ifdraft

\documentclass[a4paper,oneside,11pt]{article}
\usepackage{fullpage}

   \author{
Ilan Shomorony
  \and
A. Salman Avestimehr} 

\else

\documentclass[conference]{IEEEtran}

\author{
   \IEEEauthorblockN{Ilan Shomorony}
   \IEEEauthorblockA{Cornell University \\ is256@cornell.edu \vspace{-20mm}
   }
  \and
   \IEEEauthorblockN{A. Salman Avestimehr}
   \IEEEauthorblockA{University of Southern California \\ avestimehr@ee.usc.edu} 
  }
 
\fi

\title{A Generalized Cut-Set Bound for Deterministic Multi-Flow Networks and its Applications}

\usepackage{times}

\usepackage{amsthm}
\usepackage{amsmath}    
\usepackage{amssymb}    
\usepackage{mathrsfs}   
\usepackage{epsfig}     
\usepackage{graphicx}

\usepackage{subfigure}

\usepackage[colorlinks=true,pdfstartview=FitV,linkcolor=black,citecolor=black,urlcolor=black,plainpages=false]{hyperref}

\usepackage{amsfonts}
\usepackage{dsfont}

\usepackage[numbers,sort&compress]{natbib}

\usepackage{enumerate}

\allowdisplaybreaks[4]

\newcommand{\ep}{\epsilon}

\newcommand{\A}{{\mathcal A}}
\newcommand{\B}{{\mathcal B}}
\newcommand{\C}{{\mathcal C}}
\newcommand{\D}{{\mathcal D}}
\newcommand{\V}{{\mathcal V}}
\newcommand{\E}{{\mathcal E}}
\newcommand{\I}{{\mathcal I}}
\newcommand{\N}{{\mathcal N}}
\newcommand{\M}{{\mathcal M}}
\newcommand{\IN}{{\mathbb N}}
\newcommand{\R}{{\mathbb R}}
\renewcommand{\O}{{\mathcal O}}

\newcommand{\T}{{\mathcal T}}
\newcommand{\U}{{\mathcal U}}
\newcommand{\W}{{\mathcal W}}

\renewcommand{\S}{{\mathcal S}}
\newcommand{\K}{{\mathcal K}}
\newcommand{\X}{{\mathcal X}}
\newcommand{\Y}{{\mathcal Y}}
\newcommand{\F}{{\mathds F}}

\newcommand{\kkk}{K \hspace{-1mm} \times \hspace{-1mm} K \hspace{-1mm} \times \hspace{-1mm} K}

\newcommand{\kkkl}[1]{K \hspace{-#1mm} \times \hspace{-#1mm} K \hspace{-#1mm} \times \hspace{-#1mm} K}

\newcommand{\rank}{{\mathsf{rank}} \,}

\newcommand{\MC}[3]{#1\leftrightarrow #2 \leftrightarrow #3}

\newcommand{\defi}{\triangleq}

\newcommand{\cut}{\Omega}
\newcommand{\cutb}{\Theta}

\newcommand{\til}{\tilde}

\newcounter{constcount}
\newcounter{numcount}
\newcommand{\eqnum}{\stackrel{(\roman{numcount})}{=}\stepcounter{numcount}}

\newcommand{\leqnum}{\stackrel{(\roman{numcount})}{\leq\;}\stepcounter{numcount}}

\newcommand{\cnt}{($\roman{numcount}$)\;\stepcounter{numcount}}

\newcommand{\rescnt}{\setcounter{numcount}{1}}

\renewcommand{\ge}{\gamma}

\newcommand{\eref}[1]{(\ref{#1})}
\newcommand{\tref}[1]{Theorem~\ref{#1}}
\newcommand{\fref}[1]{Fig.~\ref{#1}}

\newcounter{thmcnt}
  \let\Oldsection\section
  
\renewcommand{\section}{\stepcounter{thmcnt}\Oldsection}

\newtheorem{theorem}{Theorem}

\newtheorem{lemma}{Lemma}
\newtheorem{definition}{Definition}
\newtheorem{claim}{Claim}

\newtheorem{cor}{Corollary}

\newtheorem{remark}{Remark}

\newenvironment{lemmarep}[1]{\noindent {\bf Lemma #1.}\begin{it}}{\end{it}}

\newenvironment{claimrep}[1]{\noindent {\bf Claim #1.}\begin{it}}{\end{it}}

\newenvironment{correp}[1]{\noindent {\bf Corollary #1.}\begin{it}}{\end{it}}

\newcounter{examplecounter}
\newenvironment{example}
{
\stepcounter{examplecounter} {\vspace{2mm} \noindent \bf Example \arabic{examplecounter}.} \begin{it} \rm
}
{
\end{it} \vspace{0mm}}

\newcommand{\aln}[1]{\begin{align*}#1\end{align*}}

\newcommand{\al}[1]{\begin{align}#1\end{align}}

\newcommand{\non}{\nonumber \\}

\newcommand{\floor}[1]{\left\lfloor #1 \right\rfloor}

\begin{document}

\maketitle

%

\begin{abstract}
We present a new outer bound for the sum capacity of general multi-unicast deterministic networks.
Intuitively, this bound can be understood as applying the cut-set bound to \emph{concatenated} copies of the original network with a special restriction on the allowed transmit signal distributions.
We first study applications to finite-field networks, where we obtain a general outer-bound expression in terms of ranks of the transfer matrices.
We then show that, even though our outer bound is for deterministic networks, a result from \cite{deterministickkk} relating the capacity of AWGN $\kkkl{0.6}$ networks and the capacity of a deterministic counterpart allows us to establish an outer bound to the DoF of $\kkk$ wireless networks with general connectivity.
This bound is tight in the case of the ``adjacent-cell interference'' topology, and yields graph-theoretic necessary and sufficient conditions for $K$ DoF to be achievable in general topologies.
\end{abstract}

\section{Introduction}
\label{sec:intro}

Characterizing network capacity is one of the central problems in network information theory.
While this problem is in general unsolved, there has been considerable success in several different research fronts. 
For single-flow wireline networks, for example, the capacity has been characterized first in the single-unicast scenario as a result of the max-flow min-cut theorem \cite{FF56} and then in the multicast scenario \cite{ACLY00} using network coding.
Later, in \cite{ADT11}, the max-flow min-cut theorem was generalized for a class of linear deterministic networks,
which motivated the characterization of the
 capacity of single-flow wireless networks  to within a constant gap 
\cite{ADT11}. 

In the case of multi-flow networks, i.e., when there are multiple data sources, 
 most of the work has focused on single-hop interference channels, for which the capacity has been determined or approximated to within a constant gap 
 in some two-user cases 
\cite{ElGamalCosta, ETW,KhandaniWIC,GerhardWIC}, 
and the degrees of freedom (DoF) have been characterized in the general $K$-user case
\cite{CadambeJafar,MotahariRealInterference}.
%
More recently, efforts have been made towards understanding more general multi-hop multi-flow networks,
such as two-unicast networks \cite{ShenviDey,MohajerZZ,xx,dof2unicastfull,GouDoFNonLayered,WangTwoUnicast} and $K$-unicast two-hop networks \cite{dofkkk}.

A classic tool in the study of network capacity is the cut-set bound \cite{elgamalbook}.
This capacity outer bound is attractive due to its generality -- it applies to arbitrary memoryless networks -- and the fact that it is a single-letter expression.
Furthermore, 
it 
is known to be tight in multicast wireline and 
linear 
deterministic networks and within a constant gap of capacity in AWGN relay networks \cite{ADT11}.
For multi-flow networks, however, the cut-set bound is easily seen to be arbitrarily loose. 
Aside from the wireline case, where improvements over the cut-set bound are known 
\cite{KamathGNS,HKL06,ThakhorFunctional}, 
most ``non-cut-set'' bounds are tied to specific settings 
(e.g., \cite{ETW,dof2unicastfull}), 
and few general techniques are known.

In this paper, we propose a new generalization to the cut-set bound for deterministic $K$-unicast networks.
The intuition behind our bound comes from noticing that a coding scheme for a $K$-unicast network $\N$, when applied to a \emph{concatenation} of multiple copies of $\N$,
can be used to achieve the original rates while inducing essentially the same distribution on the transmit signals of each copy of $\N$.
Hence, one should be able to apply the cut-set bound to the concatenated network with a restriction on the possible transmit signal distributions.
As we show, one can in fact require the transmit signals distribution on each copy to be \emph{the same}, which can significantly reduce the values that the mutual information terms attain.
%
%

In terms of applications, we
first consider linear finite-field networks.
These networks have recently received attention
as they allow the deterministic modeling of wireless networks and can provide insights about their AWGN counterparts. 
%
Similar to the cut-set bound in \cite{ADT11}, 
we obtain a general outer-bound expression in terms of ranks of transfer matrices.
We then focus on $\kkk$ topologies.
Besides being a canonical example of $K$-unicast multi-hop networks, as recently shown in \cite{dofkkk}, they reveal 
the significant role relays can play in interference management.
For binary $\kkk$ networks, our rank-based bound yields necessary and sufficient conditions for rate $K$ to be achieved.
Furthermore, 
using a result from \cite{deterministickkk} that relates the capacity of 
$\kkk$ networks under the AWGN and the truncated deterministic models,
we obtain a bound on the DoF of $\kkk$ AWGN networks with general connectivity.
This bound is tight in the case of the $\kkkl{1}$ topology with ``adjacent-cell interference'' and allows us to establish 
graph-theoretic 
necessary and sufficient conditions for $K$ DoF to be achievable in general topologies.

\iflong
\section{Problem Setup} \label{setupsec}
\fi

We consider a general $K$-unicast memoryless network $\N$, illustrated in Fig.~\ref{netfig}.
The network consists of a set of nodes $\V$, out of which we have $K$ sources $\S = \{s_1,...,s_K\}$ and $K$ corresponding destinations $\D = \{d_1,...,d_K\}$.
\begin{figure}[b] 
     \centering
       \includegraphics[trim=0 0 0 7mm, clip=false,width=0.8\linewidth]{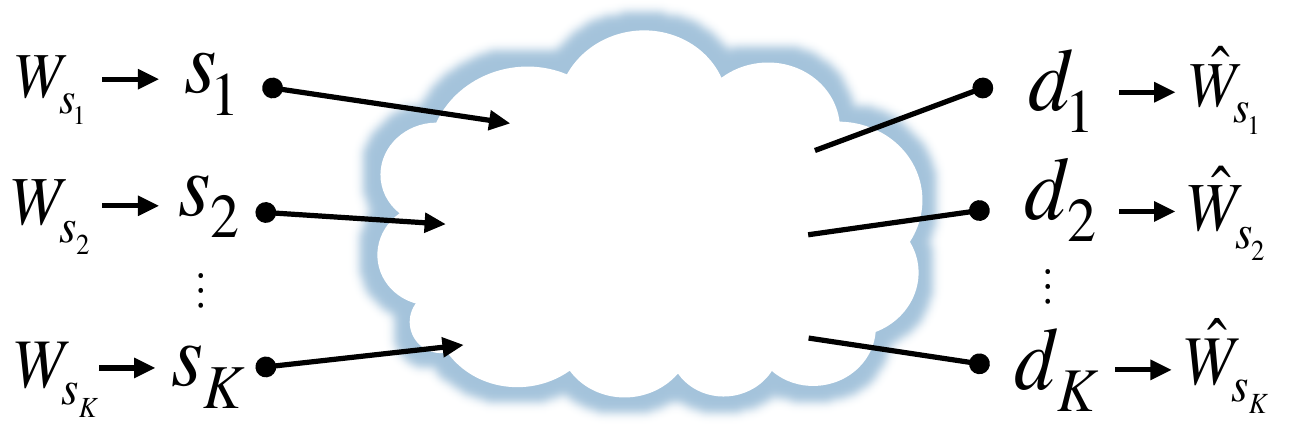} \caption{A general $K$-unicast network $\N$ \label{netfig}}
\end{figure}
At each time $t=1,2,\ldots$, each node $v \in \V$ transmits a symbol (or signal) $X_v[t] \in \X_v$ and each node $v \in \V$ receives a signal $Y_v[t] \in \Y_v$, for arbitrary alphabets $\X_v$ and $\Y_v$.
In general, for variables $z_v$ indexed by $v \in \V$, we will let $z_\A = (z_v : v \in \A )$, and for $m \geq 1$, $z_v^m = (z_v[1],...,z_v[m])$.
Then, $Y_\V[t]$, the signals received at time $t$, are determined by a function $F$ as $Y_\V[t] = F(X_\V[t])$ if $\N$ is a deterministic network or by a 
conditional distribution 
$p(y_\V|x_\V)$ if $\N$ is a stochastic (memoryless) network.
To simplify the exposition, we assume throughout that source nodes do not receive any signals 
(i.e., $\Y_{s_i} = \emptyset$)
and destination nodes do not transmit any signal 
(i.e., $\X_{d_i} = \emptyset$).


\iflong

\begin{definition} \label{codedef}
A coding scheme $\C$ with block length $n \in \IN$ and rate tuple $(R_1,\ldots,R_{K}) \in \R^{K}_+$ for a $K$-unicast network consists of
\begin{enumerate}[1. ]
\item An encoding function $f_{s_i} : \{1,\ldots,2^{n R_i}\} \to \X_{s_i}^{n}$ for each source $s_i$, $i=1,\ldots,K$
\item Relaying functions $r_{v}^{(t)} : \Y_v^{t-1} \to \X_v$, for $t=1,\ldots,n$, for each node $v \in \V \setminus (\S \cup \D)$
\item A decoding function $g_{d_i} : \Y_{d_i}^n \to \{1,\ldots,2^{n R_i}\}$ for each destination $d_i$, $i=1,\ldots,K$.
\end{enumerate}
\end{definition}

If the network imposes an average power constraint on the transmit signals, then the encoding and relaying functions above must additionally satisfy such a constraint.

\begin{definition}
The error probability of a coding scheme $\C$ (as defined in Definition \ref{codedef}), is given by
\aln{
P_{\rm error}(\C) = \Pr \left[ \bigcup_{i=1}^{K} \{ W_{s_i} \ne g_i(Y_{d_i}[1],\ldots,Y_{d_i}[n]) \} \right],
}
where we assume that each $W_{s_i}$ is chosen independently and uniformly at random from $\{1,\ldots,2^{n R_i}\}$, that source $s_i$ transmits $f_{s_i}(W_{s_i})$ over the $n$ time-steps, and node $v \in \V\setminus (\S\cup\D)$ transmits $r_v^{(t)}(Y_{v}^{t-1})$ at time $t=1,\ldots,n$, for $i=1,\ldots,K$.
\end{definition}

%
%

%
%
\begin{definition} \label{achievedef}
A rate tuple $(R_1,\ldots,R_{K})$ is said to be achievable for a $K$-unicast network if, for any $\epsilon > 0$, there exists a coding scheme $\C_n$  with rate tuple $(R_1,\ldots,R_{K})$ and some block length $n$, for which $P_{\rm error}(\C_n) \leq \epsilon$.
\end{definition}

\begin{definition}
The capacity region $C \subset \R_+^K$ of a $K$-unicast network is the closure of the set of achievable rate tuples, and the sum capacity is defined as
\aln{
C_{\Sigma} = \max_{(R_1,\ldots,R_K) \in C} \sum_{i=1}^K R_i.
}
\end{definition}

In the case of networks with an average transmit power constraint $P$, we write $C(P)$ and $C_\Sigma(P)$ for the capacity and sum capacity, and we can define the sum degrees of freedom as follows:

\begin{definition}
The sum degrees of freedom of a $K$-unicast network with transmit power constraint $P$ are defined as
\aln{
D_{\Sigma} = \lim_{P\to \infty} \frac{C_\Sigma(P)}{\tfrac12 \log P}.
}
\end{definition}

\fi  

\iflong 
\section{A Generalization of the Cut-Set Bound} 
\fi

For a $K$-unicast memoryless network, the classical cut-set bound states that, if $(R_1,...,R_K) \in C$, then there exists a distribution $p(x_\V)$ on the transmit signals of all nodes in $\V$ 
(possibly with a power constraint in the case of AWGN networks) such that
\al{
\sum_{i=1}^K R_i \leq \min_{\substack{\cut \subset \V : \\ \S \subseteq \cut \subseteq \V-\D}} I(X_\cut; Y_{\cut^c}| X_{\cut^c}). \label{cutset}
}
\iflong
This outer bound is obtained by taking a coding scheme $\C_n$ out of a sequence that achieves a rate tuple $(R_1,...,R_K)$ on $\N$ and showing that it induces a probability distribution $p(x_\V)$ on the transmit signals such that, for any cut $\cut$, the sum rate $\sum_{i=1}^K R_i$ is upper-bounded by $I(X_\cut;Y_{\cut^c}|X_{\cut^c})$ plus the Fano error term.
\fi
We generalize this bound in the case of deterministic networks as follows:

\begin{theorem} \label{genthm}
Consider a $K$-unicast deterministic network $\N$ with node set $\V$.
If a rate tuple $(R_1,...,R_K)$ is achievable on $\N$, then there exists a joint distribution $p(x_\V)$ on the transmit signals of the nodes in $\V$, such that
\al{
\sum_{i=1}^K R_i \leq \sum_{j=1}^\ell I(X_{\cut_j}; Y_{\cut_j^c}| X_{\cut_j^c}, Y_{\cut_{j-1}^c}),
\label{cutsetgen}
}
for all choices of $\ell$ node subsets $\cut_1,...,\cut_\ell$ such that $\V = \cut_0 \supseteq \cut_1 \supseteq \cut_2 \supseteq ... \supseteq \cut_\ell \supseteq \cut_{\ell+1} = \emptyset$, 
and $d_i \in \cut_j \Leftrightarrow s_i \in \cut_{j+1}$ for $j=0,1,...,\ell$, $i=1,...,K$ and any $\ell \geq 1$. 
\end{theorem}


\begin{remark}
If each $\Y_v$ is a discrete set, since the network is deterministic, the right-hand side of \eref{cutsetgen} reduces to $\sum_{j=1}^\ell H(Y_{\cut_j^c}| X_{\cut_j^c}, Y_{\cut_{j-1}^c})$.
\end{remark}

\begin{remark}
The cut-set bound in \eref{cutset} corresponds to $\ell = 1$.
\end{remark}

\begin{remark}
If the network imposes a power constraint on the transmit signals, \tref{genthm} holds for a distribution $p(x_\V)$ whose covariance matrix satisfies such a constraint.
\end{remark}

\begin{remark}
Both \eref{cutset} and \eref{cutsetgen} can be used to bound the sum of a subset of the rates by treating the remaining sources and destinations as regular nodes.
\end{remark}

\begin{remark}
In the case of wireline networks, the bound in \tref{genthm} recovers and provides an alternative interpretation to the Generalized Network Sharing (GNS) bound \cite{networksharing,KamathGNS}.
This is demonstrated in Section \ref{gnssec}.
\end{remark}

The intuition behind this bound comes from noticing that a coding scheme $\C$ designed for a network $\N$ can also be applied on a \emph{concatenation} of $\ell$ copies of $\N$, or $\N^\ell$, illustrated in Fig.~\ref{n2fig} for $\ell = 2$, obtained by identifying each destination of copies $1,2,...,\ell-1$ with the corresponding source on the next copy.
\begin{figure}[t] 
     \centering
       \includegraphics[trim=0 0 0 2mm, clip=false, width=\linewidth]{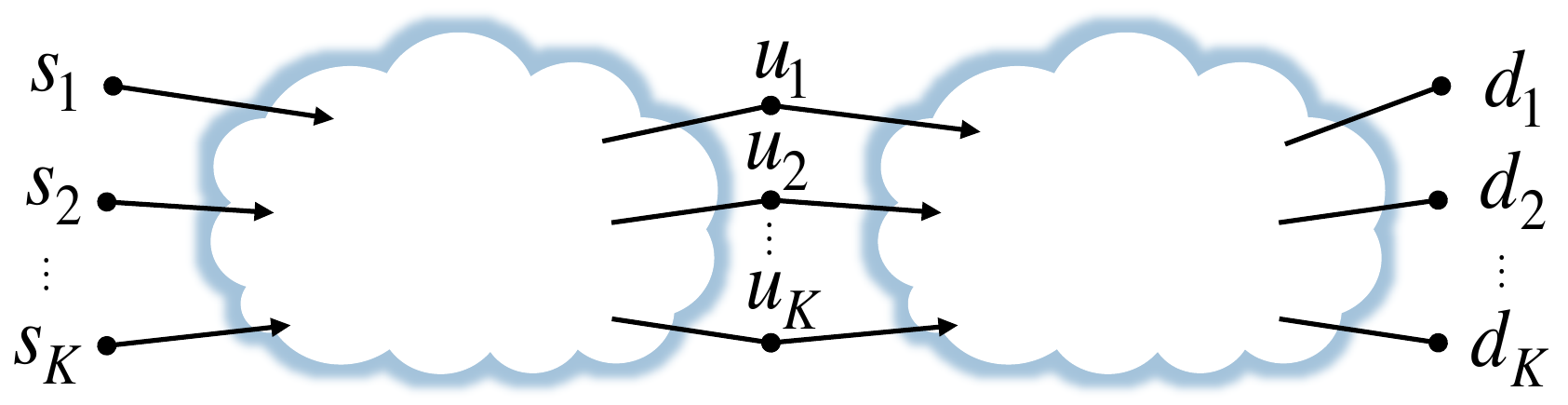} \caption{Concatenating two copies of $\N$ to obtain $\N^2$ \label{n2fig}}
\end{figure}
More precisely, we have the following claim,
whose proof, which is based on using coding scheme $\C$ on each copy of the network a repeated number of times, is presented in the the appendix.
\begin{claim} \label{claim1}
Let $C_\N$ and $C_{\N^\ell}$ be the capacity regions of a $K$-unicast memoryless network $\N$ and of the concatenation of $\ell$ copies of $\N$.
Then $C_\N \subseteq C_{\N^\ell}$.
%
%
\end{claim}
Because of Claim~\ref{claim1}, we can apply the cut-set bound to $\N^\ell$ in order to bound any sum rate achievable in 
$\N$.
Hence, 
if we let $\V_1$ and $\V_2$ be the set of nodes of the first and second copies of the network respectively (and $\V_1 \cap \V_2 = \U$), we obtain
\al{
R_\Sigma & \leq \max_{p(x_{\V_1\cup \V_2})} \min_{\cut_1,\cut_2} I(X_{\cut_1},X_{\cut_2}; Y_{\cut_1^c},Y_{\cut_2^c}| X_{\cut_1^c},X_{\cut_2^c}) 
\non
& = \max_{p(x_{\V_1\cup \V_2})} \min_{\cut_1,\cut_2} I(X_{\cut_1},X_{\cut_2}; Y_{\cut_1^c}| X_{\cut_1^c},X_{\cut_2^c}) \nonumber \\
& \quad\quad\quad\quad  + I(X_{\cut_1},X_{\cut_2}; Y_{\cut_2^c}|  X_{\cut_1^c},X_{\cut_2^c}, Y_{\cut_1^c}), 
\label{n2bound}
}
where $R_\Sigma = \sum_{i=1}^K R_i$, $\cut_i^c = \V_i \setminus \cut_i$ for $i=1,2$ and
 the minimization is over $\cut_1 \subseteq \V_1$ and $\cut_2 \subseteq \V_2 \setminus \D$ such that $\S \subseteq \cut_1$ and $\cut_1 \cap \U = \cut_2 \cap \U$. 
\iflong
We point out that this argument is tied to the multi-unicast nature of the network, which requires each $u_i$ to be individually capable of decoding its message $W_i$.
\fi
Since 
$\Y_{s_i} = \emptyset$ and $\X_{d_i} = \emptyset$ for $i=1,...,K$,
we have 
the Markov chains 
$\MC{X_{\V_2}}{X_{\V_1}}{Y_{\V_1}}$ and $\MC{X_{\V_1}}{X_{\V_2}}{Y_{\V_2\setminus\U}}$. 
Therefore, it is not difficult to see that the mutual information terms in \eref{n2bound} can be bounded by
\aln{
& I(X_{\cut_1},X_{\cut_2},X_{\cut_2^c}; Y_{\cut_1^c}| X_{\cut_1^c}) \non
& + 
 I(X_{\cut_1},  X_{\cut_1^c},X_{\cut_2}; Y_{\cut_2^c}| X_{\cut_2^c}, Y_{\cut_1^c}) \non
& =  I(X_{\cut_1}; Y_{\cut_1^c}| X_{\cut_1^c}) + 
I(X_{\cut_2},X_{\cut_2^c}; Y_{\cut_1^c}| X_{\cut_1},X_{\cut_1^c}) \non
& + \hspace{-0.5mm}
I(X_{\cut_2}; Y_{\cut_2^c}| X_{\cut_2^c}, Y_{\cut_1^c}) \hspace{-0.5mm} + \hspace{-0.5mm}
I(X_{\cut_1},  X_{\cut_1^c}; Y_{\cut_2^c}| X_{\cut_2}, X_{\cut_2^c}, Y_{\cut_1^c}) \non
& \quad \quad =  I(X_{\cut_1}; Y_{\cut_1^c}| X_{\cut_1^c}) + I(X_{\cut_2}; Y_{\cut_2^c}| X_{\cut_2^c}, Y_{\cut_1^c}),
}
and \eref{n2bound} can be written as
\al{
R_\Sigma \leq \max_{p(x_{\V_1\cup \V_2})} \min_{\cut_1,\cut_2} & I(X_{\cut_1}; Y_{\cut_1^c}| X_{\cut_1^c}) \non & 
+ I(X_{\cut_2}; Y_{\cut_2^c}| X_{\cut_2^c}, Y_{\cut_1^c}).
\label{simplegen}
}
For a general $\ell$, by following the same argument, we conclude that, if a rate tuple $(R_1,...,R_K)$ is achievable on $\N$, then there exists a joint distribution $p(x_{\V_1 \cup ... \cup \V_\ell})$ on the transmit signals of the nodes of the concatenated network $\N^\ell$, such that
\al{
R_\Sigma \leq \min_{\cut_1,...,\cut_\ell} \sum_{j=1}^\ell I(X_{\cut_j}; Y_{\cut_j^c}| X_{\cut_j^c}, Y_{1}^c,...,Y_{\cut_{j-1}^c}),
\label{simplegen}
}
where the minimization is over subsets $\cut_1,...,\cut_\ell$ such that 
$d_i \in \cut_j \Leftrightarrow s_i \in \cut_{j+1}$ for $j=0,1,...,\ell$, $i=1,...,K$.

In order to obtain a bound on $C_\Sigma$ from \eref{simplegen}, one would need to maximize the right-hand side over all joint distributions $p(x_{\V_1 \cup ... \cup \V_\ell})$.
However, as we shall see next, this maximization will result in uninteresting bounds.
First we notice that,
due to the Markov Chain $\MC{Y_{\cut_1^c},...,Y_{\cut_{j-1}^c}}{X_{\V_j}}{Y_{\Omega_j^c \setminus \cut_{j-1}^c}}$, each term in \eref{simplegen} becomes
\aln{
& I(X_{\cut_j}; Y_{\cut_j^c}| X_{\cut_j^c}, Y_{1}^c,...,Y_{\cut_{j-1}^c}) \\ 
& \hspace{-0.2cm}= I(X_{\cut_j}; Y_{\cut_j^c\setminus \cut_{j-1}^c}| X_{\cut_j^c}) - I(Y_{\cut_1^c},...,Y_{\cut_{j-1}^c}; Y_{\cut_j^c\setminus \cut_{j-1}^c}| X_{\cut_j^c}),
}
and it is not difficult to see that \eref{simplegen} is always maximized by product distributions $p(x_{\V_1\setminus\V_2})p(x_{\V_2\setminus\V_3})...p(x_{\V_\ell})$.
In the case $\ell = 2$, for example, 
(\ref{simplegen}) implies that
\al{
R_\Sigma & \leq  \max_{p(x_{\V_1 \cup \V_2})} \min_{(\cut_1,\cut_2) \in \K} {{ \Big[ }}
I(X_{\cut_1}; Y_{\cut_1^c}| X_{\cut_1^c}) \nonumber \\
& \quad \quad + I(X_{\cut_2}; Y_{\cut_2^c\setminus \U}| X_{\cut_2^c}) - I(Y_{\cut_1^c}; Y_{\cut_2^c}| X_{\cut_2^c}) {{\Big]}} \label{simp2eqa}  \\
& = \max_{p(x_{\V_1 \setminus \V_2})p(x_{\V_2})} \min_{(\cut_1,\cut_2)\in \K} {{ \Big[ }}
I(X_{\cut_1}; Y_{\cut_1^c}| X_{\cut_1^c}) \nonumber \\ 
& \quad \quad + I(X_{\cut_2}; Y_{\cut_2^c\setminus \U}| X_{\cut_2^c}) {{\Big]}} \label{simp2eqb}
}
which is similar to applying the cut-set bound first to the pairs $\{(s_i,d_i) : i \in\I\}$ where $\I = \{ i : u_i \in \U \setminus \cut_1 \}$ and then to the pairs $\{(s_i,d_i) : i \notin\I\}$ (although not exactly the same).

In \tref{genthm}, we overcome this issue by, instead of taking cuts $\cut_1 \subset \V_1,...,\cut_\ell \subset \V_\ell$ from concatenated copies of $\N$, taking multiple cuts from $\N$ itself; i.e., $\cut_j \subset \V$, for $j=1,...,\ell$ (with the additional restriction that $\cut_1 \supseteq ... \supseteq \cut_\ell$).
Thus, for deterministic networks, this can be thought of as restricting the maximization in \eref{simp2eqb} to be over distributions where $X_{\V_1} = X_{\V_2}$ with probability $1$.
Intuitively, this choice makes the negative mutual information term in \eref{simp2eqb} as large as possible.
The following example illustrates the gains of the bound in Theorem~\ref{genthm} over the traditional cut-set bound.

%



\begin{example}
Consider the binary Z-channel in \fref{za}.
It is easy to see that $C_\Sigma = 1$, while the traditional cut-set bound only implies $C_\Sigma \leq 2$.
Now consider concatenating two copies of this Z-channel and choosing $\cut_1$ and $\cut_2$ as shown in \fref{zb}.
\begin{figure}[h] 
     \centering
       \subfigure[]{
       \includegraphics[trim=0 0 0 4mm, clip=false, width=0.42\linewidth]{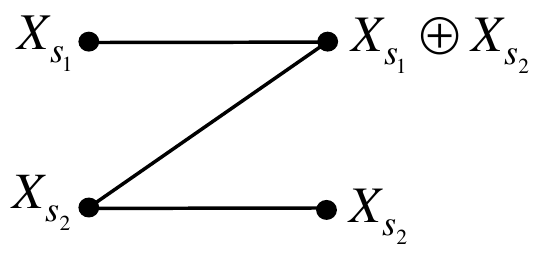} \label{za} }
    \hspace{0mm}
        \subfigure[]{
       \includegraphics[trim=0 0 0 4mm, clip=false, width=0.4\linewidth]{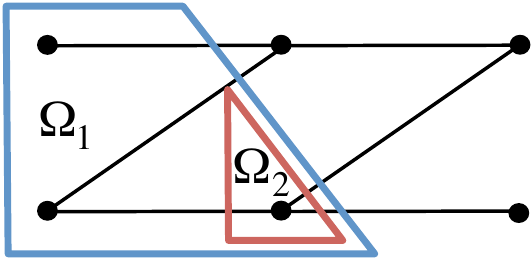} \label{zb} }  
       \caption{(a) A binary Z-channel, and (b) a possible choice of cuts for the concatenation of two binary Z-channels.
       }
\end{figure}
By maximizing over all distributions $p(x_{\V_1 \cup \V_2})$, as in \eref{simp2eqa}, we again obtain $C_\Sigma \leq 2$.
%
%
However, if we take the corresponding choices of $\cut_1$ and $\cut_2$ in \tref{genthm} (i.e., $\cut_1 = \{s_1,s_2,d_2\}$, $\cut_2 = \{s_2\}$ in the original network), we obtain
\aln{
 C_\Sigma & \leq I(X_{\cut_1}; Y_{\cut_1^c}| X_{\cut_1^c}) + I(X_{\cut_2}; Y_{\cut_2^c}| X_{\cut_2^c}, Y_{\cut_1^c}) \\
& = I(X_{s_1},X_{s_2};X_{s_1}\oplus X_{s_2}) \\
& \quad + I(X_{s_2}; X_{s_1}\oplus X_{s_2}, X_{s_2} | X_{s_1}, X_{s_1} \oplus X_{s_2})  \leq 1 + 0.
} 
\end{example}

\vspace{-3mm}


Next we prove \tref{genthm}.
Even though the motivation behind the result is based on the concatenation of multiple copies of a network $\N$, the actual proof does not involve the notion of concatenation and follows by manipulating mutual-information inequalities on the original network $\N$.

\begin{proof}[Proof of \tref{genthm}]
We first prove the case $\ell = 2$.
We let $\cut_1,\cut_2 \subset \V$ be such that $\S \subseteq \cut_1$, $\cut_2 \subseteq \cut_1 \setminus \D$ and $d_i \in \cut_1 \Leftrightarrow s_i \in \cut_{2}$ and
we let $\{\C_n\}$ be a sequence of coding schemes that achieves 
sum rate $R_\Sigma$ 
on $\N$. 
By applying coding scheme $\C_n$ of block length $n$ on $\N$, we obtain
\al{
n R_\Sigma & = H(W_\S) = I(W_\S; Y_{\D}^n) + H(W_\S|Y_{\D}^n) \non & 
\leqnum I(W_\S; Y_{\D}^n) + n \ep_n \leq  I(W_\S; Y_{\cut_2^c}^n) + n \ep_n \non
& = I(W_{\S}; Y_{\cut_2^c \cap \cut_1}^n, Y_{\cut_1^c}^n) + n \ep_n \non
& = I(W_{\S}; Y_{\cut_1^c}^n) + I(W_{\S\setminus\cut_2}, W_{\S \cap \cut_2}; Y_{\cut_2^c \cap \cut_1}^n | Y_{\cut_1^c}^n) + n \ep_n \non
& = 
\underbrace{I(W_{\S}; Y_{\cut_1^c}^n)}_{\rm I} 
+ \underbrace{I(W_{\S\setminus\cut_2}; Y_{\cut_2^c \cap \cut_1}^n | Y_{\cut_1^c}^n)}_{\rm II} \non
& \quad + \underbrace{I(W_{\S \cap \cut_2}; Y_{\cut_2^c \cap \cut_1}^n | Y_{\cut_1^c}^n, W_{\S\setminus\cut_2})}_{\rm III}
+ n \ep_n \label{termseq}
}
where \rescnt
\cnt follows from Fano's inequality. 
By following the steps in the usual cut-set bound proof (see \cite[Theorem 18.1]{elgamalbook}),  for term (I) we have
\al{
I(W_{\S}; Y_{\cut_1^c}^n) 
& = \sum_{t=1}^n I(W_{\S}; Y_{\cut_1^c}[t]| Y_{\cut_1^c}^{t-1}) \non
& = \sum_{t=1}^n  I(W_{\S}^n; Y_{\cut_1^c}[t]| Y_{\cut_1^c}^{t-1}, X_{\cut_1^c}[t]) \non
& \leq \sum_{t=1}^n  I(W_{\S}, Y_{\cut_1^c}^{t-1}; Y_{\cut_1^c}[t]|  X_{\cut_1^c}[t])  \non
& \leq \sum_{t=1}^n I(W_{\S}, Y_{\cut_1^c}^{t-1}, X_{\cut_1}[t]; Y_{\cut_1^c}[t]|  X_{\cut_1^c}[t]) \non
& \leq \sum_{t=1}^n I(X_{\cut_1}[t] ; Y_{\cut_1^c}[t]|  X_{\cut_1^c}[t]). \label{termI}
}
Term (II) can be upper-bounded by 
$H(W_{\S\setminus\cut_2} | Y_{\cut_1^c}^n) \leq n \ep'_n$, where $\ep'_n \to 0$ from Fano's inequality, since 
$s_i \in \S \setminus \cut_2 \Leftrightarrow d_i \in \cut_1^c \cap \D$. 
Finally, for term (III), we obtain \rescnt
\al{
& \hspace{-4mm} I(W_{\S \cap  \cut_2}; Y_{\cut_2^c \cap \cut_1}^n | Y_{\cut_1^c}^n, W_{\S\setminus\cut_2}) \non
& \leq I(W_{\S \cap \cut_2}; Y_{\cut_2^c}^n | Y_{\cut_1^c}^n, W_{\S\setminus\cut_2}) \non
& = \sum_{t=1}^n I(W_{\S \cap \cut_2}; Y_{\cut_2^c}[t] | Y_{\cut_2^c}^{t-1}, Y_{\cut_1^c}^n, W_{\S\setminus\cut_2}) \non
& \eqnum \sum_{t=1}^n I(W_{\S \cap \cut_2}; Y_{\cut_2^c}[t] | Y_{\cut_2^c}^{t-1}, Y_{\cut_1^c}^n, W_{\S\setminus\cut_2}, X_{\cut_2^c}[t]) \non
& \leq \sum_{t=1}^n I(W_{\S}, Y_{\cut_2^c}^{t-1}, Y_{\cut_1^c}^n, X_{\cut_2}[t]; Y_{\cut_2^c}[t] |  Y_{\cut_1^c}[t], X_{\cut_2^c}[t]) \non
& = \sum_{t=1}^n I(X_{\cut_2}[t]; Y_{\cut_2^c}[t] |  Y_{\cut_1^c}[t], X_{\cut_2^c}[t]) \non
& \quad \quad + I( W_{\S}, Y_{\cut_2^c}^{t-1}, Y_{\cut_1^c}^n; Y_{\cut_2^c}[t] | Y_{\cut_1^c}[t], X_{\V}[t]) \non
& \eqnum \sum_{t=1}^n I(X_{\cut_2}[t]; Y_{\cut_2^c}[t] |  Y_{\cut_1^c}[t], X_{\cut_2^c}[t]) \label{termIII}
}\rescnt
where
\cnt follows because from $Y_{\cut_2^c}^{t-1}$ we can build $X_{\cut_2^c \setminus \S}[t]$ and from $W_{\S \setminus \cut_2}$ we can build $X_{\cut_2^c \cap \S}[t]$ and
\cnt because $Y_{\cut_2^c}[t]$ is a function of $X_{\V}[t]$.
Therefore, \eref{termseq} implies that 
\aln{
R_\Sigma \leq & \frac1n \sum_{t=1}^n {{[}} I(X_{\cut_1}[t] ; Y_{\cut_1^c}[t]|  X_{\cut_1^c}[t]) \non
& + 
I(X_{\cut_2}[t]; Y_{\cut_2^c}[t] |  Y_{\cut_1^c}[t], X_{\cut_2^c}[t]) {{]}}
+ (\ep_n + \ep'_n).
}
Following \cite{elgamalbook}, we let $Q$ be a uniform r.v. on $\{1,...,n\}$  and we set $\tilde X_\V = X_\V[Q]$ so that $\MC{Q}{\tilde X_\V}{\tilde Y_\V}$, and we obtain
\aln{
R_\Sigma & \leq I(X_{\cut_1}[Q] ; Y_{\cut_1^c}[Q]|  X_{\cut_1^c}[Q],Q) \non
& \;\; + 
I(X_{\cut_2}[Q]; Y_{\cut_2^c}[Q] |  Y_{\cut_1^c}[Q], X_{\cut_2^c}[Q],Q) +\ep_n + \ep'_n \non
& \leq I(X_{\cut_1}[Q] ; Y_{\cut_1^c}[Q]|  X_{\cut_1^c}[Q]) \non
& \; \; + 
I(X_{\cut_2}[Q]; Y_{\cut_2^c}[Q] |  Y_{\cut_1^c}[Q], X_{\cut_2^c}[Q]) + \ep_n + \ep'_n \non
& \leq  I(\tilde X_{\cut_1} ; \tilde  Y_{\cut_1^c} | \tilde  X_{\cut_1^c} ) + 
I(\tilde X_{\cut_2}; \tilde Y_{\cut_2^c} |  \tilde Y_{\cut_1^c}, \tilde X_{\cut_2^c}) + \ep_n'',
}
where we let $\ep_n'' = \ep_n + \ep_n'$, and $\ep_n'' \to 0$ as $n \to \infty$. 
This concludes the proof in the case $\ell =2$.

Now consider the case $\ell = 3$.
Similar to the expression obtained in \eref{termseq}, this time we upper bound the sum rate as
\al{ \rescnt
n R_\Sigma & 
\leq I(W_\S; Y_{\D}^n) + n \ep_n \leq  I(W_\S; Y_{\cut_3^c}^n) + n \ep_n \non
& = I(W_{\S\setminus\cut_2}, W_{\S \cap \cut_2}; Y_{\cut_3^c \cap \cut_1}^n, Y_{\cut_1^c}^n) + n \ep_n \non
& = 
\underbrace{I(W_{\S}; Y_{\cut_1^c}^n)}_{\rm I} 
+ \underbrace{I(W_{\S\setminus\cut_2}; Y_{\cut_3^c \cap \cut_1}^n | Y_{\cut_1^c}^n)}_{\rm II} \non
& \quad + \underbrace{I(W_{\S \cap \cut_2}; Y_{\cut_3^c \cap \cut_1}^n | Y_{\cut_1^c}^n, W_{\S\setminus\cut_2})}_{\rm III}
+ n \ep_n \label{termseq2}
}
Term (I) can be bounded as in \eref{termI} and term (II) can be bounded with a Fano error term as we did for term (II) in \eref{termseq}.
Term (III) can be rewritten as
\aln{
& \hspace{-2mm} I(W_{\S \cap \cut_2}; Y_{\cut_3^c \cap \cut_1}^n | Y_{\cut_1^c}^n, W_{\S\setminus\cut_2}) \non
& = I(W_{\S \cap \cut_2}; Y_{\cut_2^c \cap \cut_1}^n, Y_{\cut_3^c \cap \cut_2}^n | Y_{\cut_1^c}^n, W_{\S\setminus\cut_2}) \non
& = \underbrace{I(W_{\S \cap \cut_2}; Y_{\cut_2^c \cap \cut_1}^n | Y_{\cut_1^c}^n, W_{\S\setminus\cut_2})}_{\rm IV} \non
& \quad + \underbrace{I(W_{\S \cap \cut_2}; Y_{\cut_3^c \cap \cut_2}^n | Y_{\cut_1^c}^n, W_{\S\setminus\cut_2}, Y_{\cut_2^c \cap \cut_1}^n)}_{\rm V}.
}
Term (IV) is the same as term (III) in \eref{termseq} and can be upper-bounded as in \eref{termIII}.
Term (V) is further broken down as
\aln{
& \hspace{-5mm} I(W_{(\S \cap \cut_2)\setminus \cut_3}, W_{\S \cap \cut_3}; Y_{\cut_3^c \cap \cut_2}^n | Y_{\cut_2^c}^n, W_{\S\setminus\cut_2}) \non
& = I(W_{(\S \cap \cut_2)\setminus \cut_3}; Y_{\cut_3^c \cap \cut_2}^n | Y_{\cut_2^c}^n, W_{\S\setminus\cut_2}) \non
& \quad \quad  + I(W_{\S \cap \cut_3}; Y_{\cut_3^c \cap \cut_2}^n | Y_{\cut_2^c}^n, W_{\S\setminus\cut_3}) \non
& \leq \underbrace{I(W_{\S \setminus \cut_3}; Y_{\cut_3^c \cap \cut_2}^n | Y_{\cut_2^c}^n, W_{\S\setminus\cut_2})}_{\rm VI} 
\non & \quad \quad 
+ \underbrace{I(W_{\S \cap \cut_3}; Y_{\cut_3^c \cap \cut_2}^n | Y_{\cut_2^c}^n, W_{\S\setminus\cut_3})}_{\rm VII}.
}
As in the case of term (II), we can upper-bound term (VI) by $H(W_{\S\setminus\cut_3} | Y_{\cut_2^c}^n) \leq n \ep''_n$, where $\ep''_n \to 0$ from Fano's inequality, since 
$s_i \in \S \setminus \cut_3 \Leftrightarrow d_i \in \cut_2^c \cap \D$.
Finally, we notice that term (VII) is exactly like term (IV) after increasing all indices by one, and can again be upper-bound as in \eref{termIII}.
By combining all these facts, from \eref{termseq2}, the sum-rate is upper-bounded as
\aln{
R_\Sigma \leq & \frac1n \sum_{t=1}^n {{[}} I(X_{\cut_1}[t] ; Y_{\cut_1^c}[t]|  X_{\cut_1^c}[t]) \non
& + 
I(X_{\cut_2}[t]; Y_{\cut_2^c}[t] |  Y_{\cut_1^c}[t], X_{\cut_2^c}[t]) \non
& + 
I(X_{\cut_3}[t]; Y_{\cut_3^c}[t] |  Y_{\cut_2^c}[t], X_{\cut_3^c}[t]) {{]}}
+ \ep'''_n.
}
Using the same time-sharing variable $Q$ as for the case $\ell = 2$, we conclude the proof for $\ell = 3$.
It is straightforward to see that similar steps can be performed for any $\ell \geq 1$.
\end{proof}

While in the wireline case, the bound in Theorem~\ref{genthm} recovers the GNS bound, its most interesting applications are in  wireless settings.
In the next section, we first consider several applications of the bound for wireless deterministic network models.
We then show how, in the wireline case, the bound reduces to the GNS bound but, under certain restrictions on the allowed coding schemes (such as linear operations), it can be used to obtain tighter bounds.

\section{Applications of the Bound}


We will first consider finite-field deterministic networks and obtain a general outer-bound expression for the sum rate. 
In the case of binary $\kkk$ networks, this bound provides
necessary and sufficient conditions for sum rate $K$ to be achievable.
We then shift our focus to two-hop AWGN networks.
Even though our outer bound only applies to deterministic networks, we will make use of a result  from \cite{deterministickkk} that relates $\kkk$ AWGN networks with a deterministic counterpart to obtain a bound for the DoF of $\kkk$ networks with arbitrary connectivity.
%
%
%
%
This bound, combined with a variation of the coding scheme introduced in \cite{dofkkk} is then used to establish necessary and sufficient conditions for $K$ DoF to be achievable on a $\kkk$ AWGN network
and to establish the DoF of the case of ``adjacent-cell interference''.

\subsection{Linear Finite-Field Networks} \label{ffsec}

A $K$-unicast linear finite-field network $\N$ is described by a directed graph $G = (\V,\E)$, where $\V$ is the node set and $\E$ is the edge set.
If the network is layered, the node set $\V$ can be partitioned into $r$ subsets $\V_1,\V_2,...,\V_r$ (the layers) in such a way that $\E \subset \bigcup_{i=1}^{r-1} \V_i \times \V_{i+1}$, and $\V_1 = \S = \{s_1,...,s_K\}$, $V_r = \D = \{d_1,...,d_K\}$.  
To each edge $(u,w) \in \E$ we associate a nonzero channel gain $F(u,w)$ from a given finite field $\F$.
For two sets of nodes 
$\U \subseteq \V_i$ and $\W \subseteq \V_{i+1}$, we let $F(\U,\W)$ be the $|\W| \times |\U|$ transfer matrix 
from $\U$ to $\W$.
The received signals at layer $\V_{j+1}$ are given by $Y_{\V_{j+1}}[t] = F(\V_j,\V_{j+1}) X_{\V_j}[t]$ for $t=1,...,n$.
For conciseness, we let $\rank (\U;\W) \defi \rank F(\U,\W)$, and for $\cut \subset \V$, we let $\cut[j] = \cut \cap \V_j$.
We also let $\bar C_\Sigma = C_\Sigma / \log|\F|$ be the normalized sum capacity.
We have the following corollary of \tref{genthm}.

\begin{cor} \label{ffcor}
For a layered $K$-unicast linear finite-field network $\N$ as described above,
if $R_\Sigma \leq \bar C_\Sigma$, we must have 
\al{
& R_\Sigma \leq \sum_{j=1}^{r-1} \rank(\cut[j]; \cut^c[j+1])
+ \rank(\cutb[j]; \cutb^c[j+1]) \non 
& \hspace{20mm} - \rank( \cutb[j] ; \cut^c[j+1])
\label{ffgen}
}
for any node subsets $\cut$ and $\cutb$ such that $\cutb \subset \cut \setminus \D$, $\S \subset \cut$ and 
$d_i \in \cut \Leftrightarrow s_i \in \cutb$. 
\end{cor}

\iflong

\begin{proof}
We apply \tref{genthm} with $\cut_1 = \cut$ and $\cut_2 = \cutb$.
For the first term in the sum in \eref{cutsetgen}, we have
\aln{
H(Y_{\cut^c}| X_{\cut^c}) \leq \sum_{j=1}^{r-1} \rank(\cut[j]; \cut^c[j+1]) \cdot \log |\F|,
}
and for the second term we have 
\aln{ \rescnt
& H(Y_{\cutb^c}  |  X_{\cutb^c},Y_{\cut^c})  
\leq \sum_{j=1}^{r-1} H(Y_{\cutb^c[j+1]}|X_{\cutb^c[j]}, Y_{\cut^c[j+1]}) \\
& \leq \sum_{j=1}^{r-1} H(F(\cutb[j];\cutb^c[j+1])X_{\cutb[j]}| F(\cutb[j];\cut^c[j+1])X_{\cutb[j]}) \\
& \leqnum \sum_{j=1}^{r-1} {{\Big(}} \rank(\cutb[j]; \cutb^c[j+1]) \non
& \quad \quad \quad \quad  - \rank( \cutb[j] ; \cut^c[j+1]) {{\Big)}}\cdot \log |\F|,
} \rescnt
where \cnt follows since $H(A {\bf x} | B{\bf x}) / \log|\F| \leq \rank \begin{bmatrix} A \\ B \end{bmatrix}  - \rank B$ 
\iflong
from Lemma \ref{ranklem} in the appendix.
\else
(see \cite{gencutset} for details).
\fi
\end{proof}

\fi

We point out that it is straightforward to generalize Corollary \ref{ffcor} to the wireless deterministic network model from \cite{ADT11} or to general finite-field networks with MIMO nodes.

We now shift our focus to $\kkk$ networks, 
%
i.e., when $r = 3$ and $\V_2 = \{u_1,...,u_K\} \defi \U$.
This network was recently studied in the AWGN case in \cite{dofkkk}, where $K$ DoF were shown to be achievable.
This result suggested that significant gains can be obtained from two-hop interference management, raising interest in the study of different two-hop network models.
The following result provides necessary conditions for sum rate $K$ to be achieved in the finite-field case.

\begin{cor} \label{ffcor2}
For a $\kkk$ finite-field network, if $\bar C_\Sigma = K$, then 
$F(\U,\D)$ and $F(\S,\U)$ must be invertible and 
\begin{enumerate}[(i)]
\item $F(s_i,u_j) = 0$ if and only if $\det F(\U \setminus \{u_j\},\D \setminus \{d_i\}) = 0$, for any $i,j$
\item $F(u_j,d_i) = 0$ if and only if $\det F(\S \setminus \{s_i\},\U \setminus \{u_j\}) = 0$, for any $i,j$.
\end{enumerate}
Otherwise, 
$\bar C_\Sigma \leq K-1$.
\end{cor}

\iftrue

\begin{proof}
Clearly, if $F(\U,\D)$ or $F(\S,\U)$ are not invertible, $\bar C_\Sigma \leq K-1$.
We consider applying Corollary~\ref{ffcor} with four different choices of $\cut$ and $\cutb$.
For $\cut = \S \cup \U \cup \{d_i\}$ and $\cutb = \{s_i\} \cup (\U\setminus\{u_j\})$, 
if $\bar C_\Sigma = K$,
we obtain
\aln{
& K = \bar C_\Sigma 
\leq \rank(\S; \emptyset) + \rank(\{s_i\}; \{u_j\}) - \rank( \{s_i\} ; \emptyset) \non
& \quad \quad \quad \quad \quad + \rank(\U;\D \setminus \{d_i\}) + \rank(\U \setminus\{u_j\}; \D) \non
& \quad \quad \quad \quad \quad  - \rank( \U \setminus\{u_j\}; \D \setminus \{d_i\}) \non
& \leq 2(K-1)  + \rank(\{s_i\};\{u_j\}) - \rank( \U \setminus \{u_j\}; \D\setminus\{d_i\}) }
which implies 
\al{ \label{ffproof1}
\rank( \U \hspace{-0.5mm}\setminus\hspace{-0.5mm} \{u_j\}; \D\setminus\{d_i\}) - \rank(\{s_i\};\{u_j\}) \leq K\hspace{-0.5mm}-\hspace{-0.3mm}2.
}
Next, by choosing
$\cut = \S  \cup \U \setminus \{u_j\} \cup \{d_i\}$ and $\cutb = \{s_i\}$
$\cut = \S  \cup \{u_j\} \cup \D \setminus \{d_i\}$ and $\cutb = \S \setminus \{s_i\}$, and $\cut = \S  \cup \U \cup \D \setminus \{d_i\} $ and $\cutb = \S \setminus \{s_i\} \cup \{u_j\}$
we respectively obtain
\al{ 
& K \hspace{-0.5mm}-\hspace{-0.3mm} 2 \leq \rank( \U\hspace{-0.5mm} \setminus \hspace{-0.5mm}\{u_j\}; \D\setminus\{d_i\}) - \rank(\{s_i\};\{u_j\}), \label{ffproof2} 
\\
& \rank( \S \setminus \{s_i\}; \U\hspace{-0.5mm}\setminus\hspace{-0.5mm}\{u_j\}) - \rank(\{u_j\};\{d_i\}) \leq K\hspace{-0.5mm}-\hspace{-0.3mm}2, \label{ffproof3}\\
& K \hspace{-0.5mm}-\hspace{-0.3mm} 2 \leq \rank( \S \setminus \{s_i\}; \U\hspace{-0.5mm}\setminus\hspace{-0.5mm}\{u_j\}) - \rank(\{u_j\};\{d_i\}). \label{ffproof4}
}
Combining \eref{ffproof1} and \eref{ffproof2}, we conclude 
that \eref{ffproof2} holds with equality.
Since $F(\U,\D)$ is invertible, by Lemma \ref{submatrixlem}, $\rank( \U \setminus \{u_j\}; \D\setminus\{d_i\}) = K-2$ if 
$\det F(\U \setminus \{u_j\},\D \setminus \{d_i\}) = 0$ and $\rank( \U \setminus \{u_j\}; \D\setminus\{d_i\}) = K-1$ otherwise, implying ($i$).
Similarly, \eref{ffproof3}, \eref{ffproof4} and the fact that $F(\S,\U)$ is invertible imply ($ii$).
\end{proof}

\else

The proof consists of applying Corollary~\ref{ffcor} with four different choices of $\cut$ and $\cutb$.
See \cite{gencutset} for details.
\fi
In the case $\F = GF(2)$, the conditions in Corollary \ref{ffcor2} 
are in fact sufficient, and they imply the following:

\begin{cor} \label{ffcor3}
For a $\kkk$ finite-field network with $\F = GF(2)$, $C_\Sigma = K$ if and only if $F(\U,\D) F(\S,\U) = I$.
\end{cor}

\begin{proof}
Clearly, if $C_\Sigma = K$, $F(\S,\U)$ must be invertible. 
When $\F = GF(2)$, condition $(ii)$ in Corollary \ref{ffcor2} is equivalent to $\det F(\S \setminus \{s_i\},\U\setminus\{u_j\}) = F (u_j,d_i)$.
By definition, the $(i,j)$th entry of $F(\S,\U)^{-1}$ can be written as the $(j,i)$th cofactor of $F(\S,\U)$ divided by $\det F(\S,\U) = 1$, i.e., 
\aln{ \rescnt
\left[F(\S,\U)^{-1} \right]_{i,j} & = \frac{ \det F(\S\setminus \{s_i\},\U\setminus\{u_j\})}{\det F(\S,\U)} \non
& = F (u_j,d_i) = \left[ F(\U,\D) \right]_{i,j},
} \rescnt
and we conclude that $F(\U,\D)F(\S,\U)=I$.
Obviously, in this case, sum rate $K$ can be achieved by having each relay forward its received signal. %
\end{proof}

\subsection{Two-hop AWGN Networks}

In this section we focus on $\kkkl{0.5}$ wireless networks under an AWGN channel model.
We follow the setup in Section \ref{ffsec}, except that $\F = \R$, 
\al{
\textstyle{Y_{v}[t] = \sum_{u \in \V} F(u,v) X_{u}[t] + Z_{v}[t]}
}
is the received signal at node $v \in \V \setminus \S$ at time $t$,
where $Z_{v}[t]$ is the usual additive white Gaussian noise process,
and there is a transmit power constraint $E[X_v^2] \leq P$ for $v \in \V \setminus \D$.
We will also consider the truncated deterministic channel model \cite{ADT11}, where 
we still have a power constraint on $X_v$, but
\al{
\textstyle{Y_{v}[t] = \left\lfloor \sum_{u \in \V} F(u,v) X_{u}[t] \right\rfloor},
}
is the received signal.
Based on the characterization of the Gaussian noise as the worst-case additive noise for wireless networks in \cite{wcnoisefull}, 
the following result was established in \cite{deterministickkk}, relating the sum DoF, $D_\Sigma$, 
under these two models:

\begin{lemma}[{{\cite[Corollary 2]{deterministickkk}}}] \label{truncatedlem}
If $F(\U,\D)$ and $F(\S,\U)$ are invertible, the sum DoF of the $\kkk$ wireless network under the AWGN channel model and under the truncated channel model satisfy $D_{\Sigma,{\rm AWGN}} \leq D_{\Sigma,{\rm Truncated}}$.
\end{lemma}

Because of Lemma~\ref{truncatedlem}, any upper bound 
for the sum DoF of a $\kkk$ network (with invertible transfer matrices) under the truncated model is also a bound for the sum DoF of the corresponding AWGN network.
Since the $\kkk$ network under the truncated channel model is a deterministic network, we can use \tref{genthm} to upper-bound $C_\Sigma$ and also $D_\Sigma$.
Moreover, as implied by \cite[Lemma 7.2]{ADT11}, the DoF of a MIMO channel under the truncated deterministic model are given by the rank of the channel matrix.
We obtain a version of Corollary~\ref{ffcor} for truncated deterministic networks:

\vspace{2mm}

\begin{correp}{\ref{ffcor}'}
For a layered $K$-unicast truncated deterministic network $\N$, we must have 
\aln{ 
& D_\Sigma \leq \sum_{j=1}^{r-1} \rank(\cut[j]; \cut^c[j+1])
+ \rank(\cutb[j]; \cutb^c[j+1]) \non 
& \hspace{20mm} - \rank( \cutb[j] ; \cut^c[j+1])
}
for any node subsets $\cut$ and $\cutb$ such that $\cutb \subset \cut \setminus \D$, $\S \subset \cut$ and 
$d_i \in \cut \Leftrightarrow s_i \in \cutb$. 
\end{correp}

\begin{proof}
This result follows using the same steps as in the proof of Corollary \ref{ffcor}, except that,
instead of Lemma \ref{ranklem}, we use Lemma \ref{truncranklem}, which is based on \cite[Lemma 7.2]{ADT11}.
\end{proof}

Since Corollary \ref{ffcor2} follows directly from Corollary \ref{ffcor}, we can also replace $\bar C_\Sigma$ with $D_\Sigma$ in Corollary \ref{ffcor2} and obtain necessary conditions for $K$ DoF to be achievable in a truncated deterministic $\kkk$ network. 
By Lemma \ref{truncatedlem}, these conditions are also necessary in the case of AWGN networks,
%
%
and interestingly, they turn out to also be sufficient.
We will say that two node sets $\A, \B \subset \V$ are matched if 
there is a perfect matching between $\A$ and $\B$ in $\E$.
Then we have:

\begin{theorem} \label{dofthm}
For a $\kkk$ AWGN network, if $F(\U,\D)$ and $F(\S,\U)$ are invertible and 
\begin{enumerate}[(i)]
\item $(s_i,u_j) \in \E \Longleftrightarrow \U \setminus \{u_j\}$ and $\D \setminus \{d_i\}$ are matched, 
\item $(u_j,d_i) \in \E \Longleftrightarrow \S \setminus \{s_i\}$ and $\U \setminus \{u_j\}$ are matched
\end{enumerate}
for any $i,j$, then, for almost all values of channel gains (of existing edges), $D_\Sigma = K$.
Otherwise, 
$D_\Sigma \leq K-1$ for almost all values of channel gains.
\end{theorem}
\iffalse

\begin{proof}[Proof of Converse]
If $\U \setminus \{u_j\}$ and $\D \setminus \{d_i\}$ are not matched, for any choice of channel gains, we will have $\det F(\U \setminus \{u_j\},\D \setminus \{d_i\}) = 0$.
Since $(s_i,u_j) \notin \E$ implies $F(s_i,u_j) = 0$, we see that, if condition ($i$) above is not satisfied, condition  ($i$) in Corollary~\ref{ffcor2} is not satisfied.
Similarly, if condition ($ii$) above is not satisfied, condition  ($ii$) in Corollary~\ref{ffcor2} is not satisfied.
\end{proof}

\else
The necessary part follows by the previous discussion and by noticing that, for almost all choices of channel gains, ($i$) and ($ii$) are equivalent to ($i$) and ($ii$) in Corollary \ref{ffcor2}.
\fi
In order to prove the achievability part, we first need the following definition and lemma.


\begin{definition} \label{diagdef}
A $\kkk$ network with edge set $\E$ is diagonalizable if, for almost all assignments of real-valued channel gains to edges in $\E$, 
\begin{itemize}
\item $F(\S,\U)^{-1}$ and $F(\U,\D)$ have zeros at the same entries,
\item $F(\U,\D)^{-1}$  and $F(\S,\U)$ have zeros at the same entries.
\end{itemize}
\end{definition}

Whereas the Aligned Network Diagonalization (AND) scheme was introduced in \cite{dofkkk} for the case of $\kkk$ networks with fully connected hops, it can be extended to the class of diagonalizable networks.
This implies the following lemma, which we prove in the appendix.
%

\begin{lemma} \label{andlem}
If a $\kkk$ AWGN network is diagonalizable, then for almost all values of the channel gains, $D_\Sigma = K$.
\end{lemma}

This lemma allows us to complete the proof of \tref{dofthm}.

\begin{proof}[Proof of Achievability of \tref{dofthm}]
The $(i,j)$th entry of $F(\S,\U)^{-1}$ can be written as 
\aln{
[F(\S,\U)^{-1}]_{i,j} = \frac{\det (\S\setminus\{s_i\},\U\setminus\{u_j\})}{\det F(\S,\U)}.}
Therefore, $[F(\S,\U)^{-1}]_{i,j}$ is nonzero if and only if $\det (\S\setminus\{s_i\},\U\setminus\{u_j\})$ is nonzero.
The latter occurs for almost all values of channel gains if and only if $\S \setminus \{s_i\}$ and $\U \setminus \{u_j\}$ are matched, which by ($ii$) occurs if and only if $F(u_j,d_i) = [F(\U,\D)]_{i,j} \ne 0$.
Analogously we conclude that, for almost all values of channel gains, $[F(\U,\D)^{-1}]_{i,j}$ is nonzero if and only if $[F(\S,\U)]_{i,j}$ is nonzero.
Thus if a $\kkkl{0.5}$ AWGN network satisfies the conditions in \tref{dofthm}, it is diagonalizable and by Lemma~\ref{andlem}, $K$ DoF are achievable for almost all values of channel gains.
%
%
\end{proof}

\subsection{Two-Hop Networks with Adjacent-Cell Interference}

The bound from Corollary~\ref{ffcor}, when applied to the DoF of $\kkk$ AWGN networks, is also tight for the case of ``adjacent-cell interference''.  
As illustrated in \fref{adjfig},
\begin{figure}[ht] 
     \centering
       \includegraphics[trim=0 2mm 0 3mm, clip=false, width=0.5\linewidth]{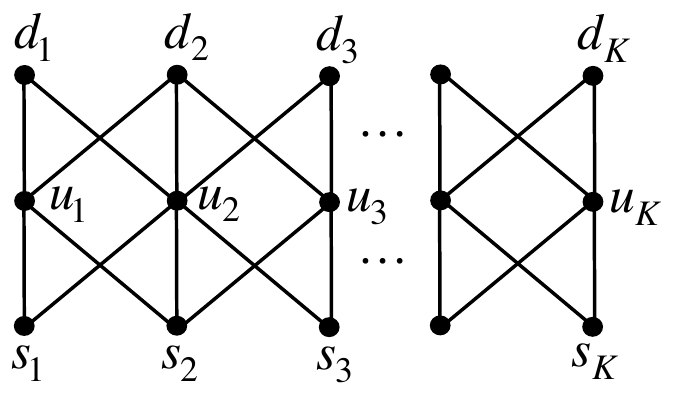} \caption{The $\kkk$ Wireless Network with adjacent-cell interference. \label{adjfig}}
\end{figure}
for this class of networks, $\E = \{ (s_i,u_j) : |i-j|\leq1\} \cup \{ (u_i,d_j) : |i-j|\leq1\}$.
This configuration is motivated in the literature as the result of two-hop communication within each cell, when interference only occurs between adjacent cells \cite{SimeoneMesh}.

\begin{theorem} \label{adjacentthm}
The AWGN $\kkkl{0.5}$ adjacent-cell interference network has $\left\lceil \frac{2K}{3} \right\rceil$ DoF for almost all values of channel gains.
\end{theorem}

\begin{proof}
For the achievability, we consider several $2 \times 2 \times 2$ subnetworks formed by $\{s_i,s_{i+1}, u_i, u_{i+1}, d_i, d_{i+1}\}$ for $i=1,4,7,...$.
In each one, we can use \cite{xx} to achieve $2$ DoF, leaving the remaining nodes as ``buffers'' to prevent any interference between differen $2 \times 2 \times 2$ channels.
If $K = 1+3m$ for some $m \in \IN$, we utilize $\{s_K,u_K,d_K\}$ as a linear network where $1$ DoF can be achieved.
It is not difficult to see that this scheme achieves $\left\lceil \frac{2K}{3} \right\rceil$ DoF.

%
%

For the converse, we use the bound from Corollary~\ref{ffcor}, with 
$
\cut  = \S \cup \U \cup d_\B \;\; \text{and} \;\; \cutb = s_\B \cup u_\A$, 
where
$\A = \left( \{1,2\} \cup \{5,6,7,8\} \cup \{11,12,13,14\}  ... \right) \cap \K$,
$\B  = \left(\{1\} \cup \{6,7\} \cup \{12,13\} \cup ... \right) \cap \K$
and $\K = \{1,...,K\}$.
First we notice that, since no index in $\B$ is adjacent to an index in $\A^c$, we have 
$\rank (s_\B;u_{\A^c})=0$.
Moreover, we have
\aln{
& \rank(\U; d_{\B^c}) + \rank(u_\A; \D) \leq |\A| + |\B^c| \non
& = \left| \left( \{1,2\} \cup \{5,6,7,8\} \cup \{11,12,13,14\} \cup ... \right) \cap \K \right| \\
& \; + \left| \left( \{2,3,4,5\} \cup \{8,9,10,11\} \cup \{14,15,16,17\} \cup ... \right) \cap \K \right| \\
& = K + \left| \{2,5,8,11,...\} \cap \K \right| = K + \left\lfloor (K+1)/3 \right\rfloor.
}
In order to compute $\rank (u_{\A},d_{\B^c})$, we notice that with the nodes of $u_{\A}$ and $d_{\B^c}$ we can build the matching
\aln{
\{ (1,2),\hspace{-0.5mm}(2,3),\hspace{-0.5mm}(5,4),\hspace{-0.5mm}(6,5),\hspace{-0.5mm}(7,8),\hspace{-0.5mm}(8,9),\hspace{-0.5mm}(11,10),... \} \cap \K\times \K,
}
which can be verified to have cardinality $\left\lceil 2(K-1)/3 \right\rceil$.
Since either all the nodes in $u_\A$ or all the nodes in $d_{\B^c}$ are in this matching, we conclude that $\rank (u_{\A},d_{\B^c}) = \left\lceil 2(K-1)/3 \right\rceil$ for almost all values of channel gains, and the bound in \eref{ffgen} reduces to
\aln{
K + \left\lfloor (K+1)/3\right\rfloor - \left\lceil 2(K-1)/3 \right\rceil = \left\lceil 2K/3 \right\rceil,
}
as we wanted to show.
\end{proof}


\iflong

\subsection{Alternative Interpretation of the GNS Bound} \label{gnssec}


A $K$-unicast wireline network $\N$ is characterized by a directed acyclic graph $G(\V,\E)$ where $\V$ is the node set and $\E$ the edge set. 
We let $\I(v) = \{ u : (u,v) \in \E\}$ and $\O(v) = \{u : (v,u) \in \E\}$ and $\Delta = \max_v \max(|\O(v)|,|\I(v)|)$.
At each time $t$, each $v \in \V$ transmits a vector $X_v[t] \in \F^{|\O(v)|}$, for some finite field $\F$, where each component is called 
$X_{v,u}[t]$ for some $u \in \O(v)$. 
Each $v \in \V$ receives a vector $Y_v[t] \in \F^{|\I(v)|}$, whose components  are $X_{u,v}[t]$ for $u \in \I(v)$.

For $K$-unicast wireline networks, we consider a special case of the bound in \tref{genthm} that can be seen to be equivalent to the Generalized Network Sharing bound \cite{networksharing,KamathGNS} and presents an alternative interpretation of this bound.

\begin{cor}[GNS Bound] \label{gnsthm}
Let $\N^\ell$ be the concatenation of $\ell$ copies of a $K$-unicast wireline network $\N$.
Suppose there is a set of edges $\M \subset \E$ such that, by removing $\M$ from each of the $\ell$ copies of $\N$ in $\N^\ell$, all sources and destinations in $\N^\ell$ are disconnected.
Then any rate tuple $(R_1,...,R_K) \cdot \log |\F|$ achievable on $\N$ must satisfy
\al{ \label{gns}
\sum_{i=1}^K R_i \leq |\M|.
}
\end{cor}

\begin{proof}
Let $\cut$ be the set of nodes in $\N^\ell$ that are reachable from a source through a path that does not contain any edges in any of the copies of $\M$.
Now let $\cut_i$ be the nodes in $\cut$ that are in the $i$th copy of $\N$.
It is not difficult to check that $\cut_1,...,\cut_\ell$ satisfy the conditions of \tref{genthm}.
Now let $\delta(A,B) = \{(u,v) \in \E : u \in \A, v \in \B\}$.
We notice that if $v \in \cut_j^c \setminus \S$ for some $j$, for each $u \in \I(v)$ we must either have $u \in \cut_j^c$ or $(u,v) \in \M \cap \delta(\cut_j,\cut_j^c)$ (or else $v$ would be in $\cut_j$).
Hence, 
\aln{
H(Y_{\cut_j^c}| X_{\cut_j^c}, Y_{\cut_{j-1}^c}) & \leq H(Y_{\cut_j^c \cap \cut_{j-1}}| X_{\cut_j^c}) \non
& \hspace{-1cm} =H(X_{u,v} : (u,v) \in \M \cap \delta(\cut_j,\cut_j^c \cap \cut_{j-1}) ) 
}
Finally, since the sets $\cut_1^c$, $\cut_2^c \cap \cut_1$,..., $\cut_\ell^c \cap \cut_{\ell-1}$, are pairwise disjoint, \tref{genthm} implies that
\aln{
\sum_{i=1}^K R_i & \leq \sum_{j=1}^\ell H(X_{u,v} : (u,v) \in \M \cap \delta(\cut_j,\cut_j^c \cap \cut_{j-1}) ) 
}
%
%
%
%
%
%
\hspace{1.25cm}  $\leq |\M| \log |\F|.$ 
\end{proof}


It is easy to check that this bound is equivalent to the GNS bound as stated in \cite{KamathGNS2}.
Moreover, the conditions in Corollary~\ref{gnsthm} provide a new interpretation to the bound,  illustrated in \fref{gnsfig}.
\begin{figure}[h] 
     \centering
       \includegraphics[trim=0 0 0 4mm, clip=false,width=0.8\linewidth]{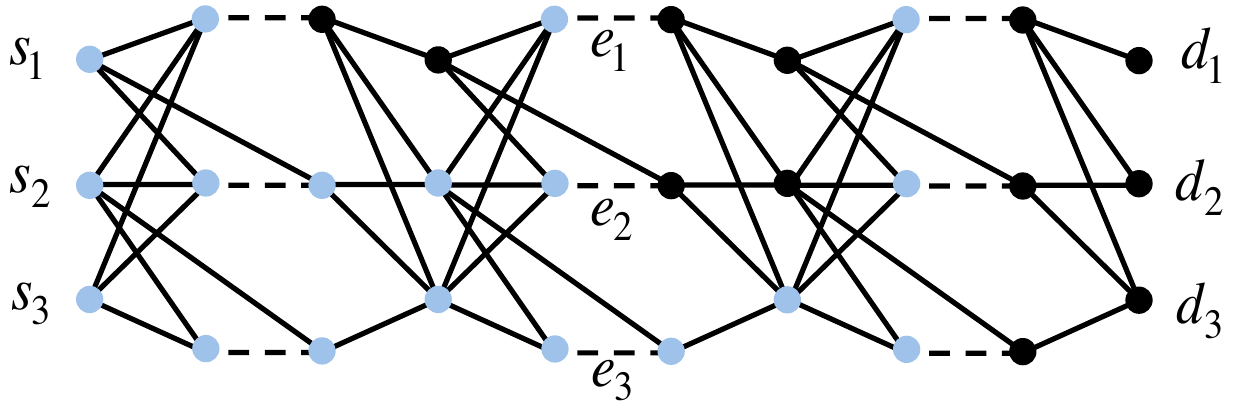} \caption{Illustration of the GNS bound for a $3$-unicast network $\N$. By removing edges $e_1$, $e_2$ and $e_3$ (dashed) from all three  copies of $\N$, we disconnect all sources and  destinations.
The nodes that can be reached from the sources after removing $e_1$, $e_2$ and $e_3$ (in blue) form $\cut_1$, $\cut_2$ and $\cut_3$.\label{gnsfig}}
\end{figure}

\subsection{Bounds for Linear Network Coding}

Since the bound in \tref{genthm} holds for general deterministic networks, if one restricts the kinds of relaying operations that can be used (say, to linear), these operations can be absorbed into the network.
In this section, we illustrate one such example, where \tref{genthm} can be used to obtain a bound that is tighter than the GNS bound.
Consider the wireline network in \fref{kamathfig}, first introduced in \cite{KamathGNS}.
\begin{figure}[ht] 
     \centering
       \includegraphics[width=0.7\linewidth]{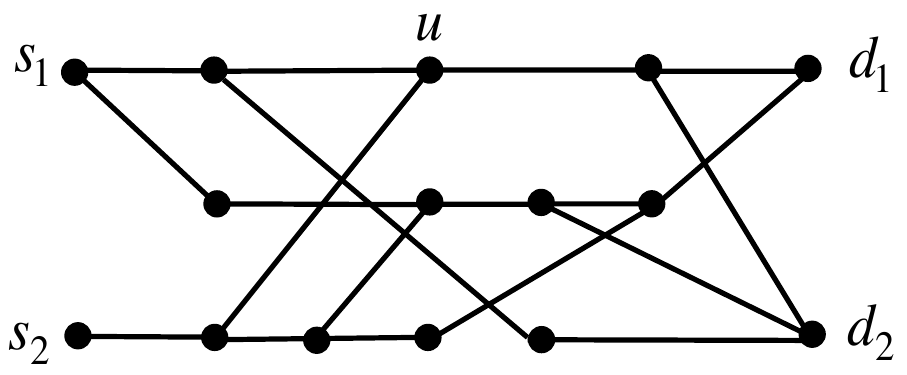} \caption{Two-unicast network where the GNS bound is not tight \cite{KamathGNS}\label{kamathfig}}
\end{figure}
With the purpose of finding an upper bound on $2 R_1 + R_2$, we consider applying the concept of network concatenation but this time in a different fashion.
We will concatenate the network in \fref{kamathfig} sideways, as shown in \fref{kamathcuts}.
It is not difficult to see that if $(R_1,R_2)$ is achieved in the network in \fref{kamathfig}, then we can achieve rate $(R_1,R_2,R_1)$ in this new network.
Moreover, the fact that \tref{genthm} can be applied to general deterministic networks  allows us to consider a mixed network model where the nodes in the junction must ``broadcast'' the same signal into both copies of the network.
Moreover, we can remove the dashed edges, since $d_2$ should be able to decode its message just using the signals from the first copy of the network.
\begin{figure}[ht] 
     \centering
       \includegraphics[width=0.7\linewidth]{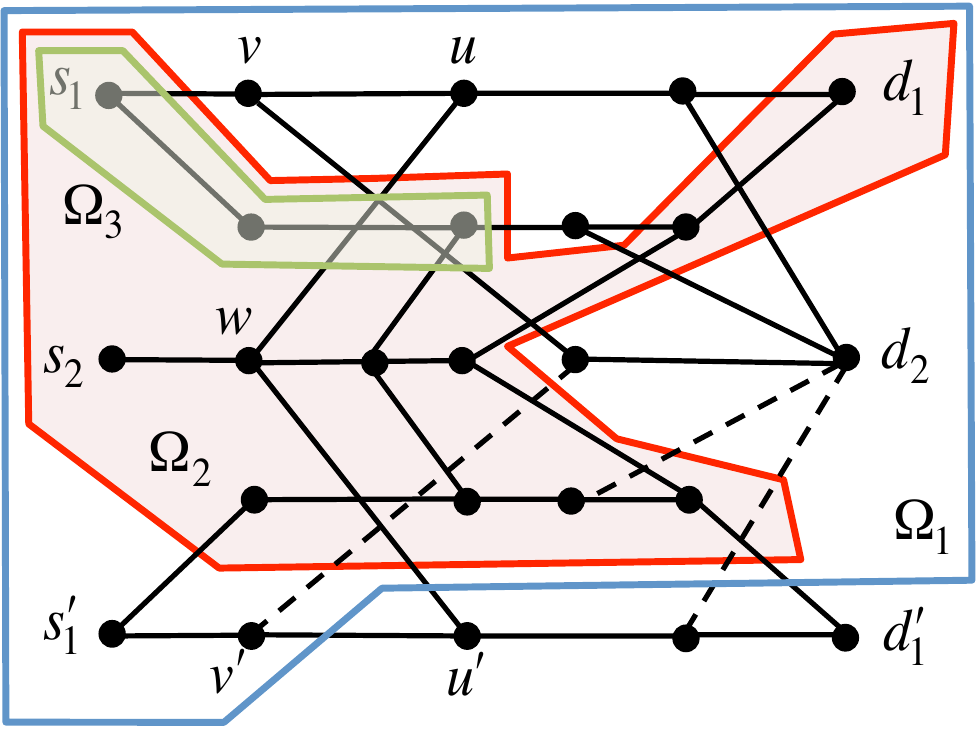} \caption{Cut choices to obtain a bound on $2R_1 + R_2$\label{kamathcuts}}
\end{figure}
Now applying \tref{genthm} with $\cut_1, \cut_2$ and $\cut_3$ as shown in \fref{kamathcuts}, we have
\aln{
2 R_1 + R_2 & \leq 2 + [2 + H(Y_{u} | X_v, Y_{u'}, X_{v'})] + 0.
}
Finally we notice that, if we restrict ourselves to linear network coding, we can absorb the operation performed at $u$ into the network, and have $Y_u$ be the result of this operation.
In this case, since $X_{w,u} = X_{w,u'}$, we will have $H(Y_{u} | X_v, Y_{u'}, X_{v'}) = 0$, which implies $2R_1 + R_2 \leq 4$.
By noticing that $R_2 \leq 1$, this implies $R_1+R_2 \leq 2.5$, which is achievable by linear network coding, as shown in \cite{KamathGNS}.

\fi

\iflong

\appendix

\begin{claimrep}{\ref{claim1}}
Let $C_\N$ and $C_{\N^\ell}$ be the capacity regions of a $K$-unicast memoryless network $\N$ and of the concatenation of $\ell$ copies of $\N$.
Then $C_\N \subseteq C_{\N^\ell}$.
\end{claimrep}
\begin{proof}[Proof]
We prove the case $\ell = 2$.
The general case follows similarly.
We show that if $(R_1,...,R_K)$ is achievable in $\N$, then $(R_1(1-\delta),...,R_K(1-\delta))$ is achievable in $\N^2$ for any $\delta > 0$.
Consider a coding scheme $\C_n$ for $\N$ with rate tuple $(R_1,...,R_K)$ and error probability $P_{\rm error}(\C_n) =\ep_n$.
For an arbitrary $\delta > 0$, we construct a new coding scheme with rate tuple $(R_1(1-\delta),...,R_K(1-\delta))$ and block length $n L$ for the concatenated network $\N^2$, where we let $L = \lfloor \ep^{-1/2}_n \rfloor$, as follows.
Each source $s_i$ will view its message $W_i \in \{1,...,2^{nL R_i(1-\delta)}\}$ as $L(1-\delta)$ messages $W_i^{(1)},...,W_i^{(L(1-\delta))}$ in $\{1,..., 2^{nR}\}$.
In the $j$th block of length $n$, the sources and relays in the first copy of $\N$ behave as if they were simply using coding scheme $\C_n$ with messages $W_1^{(j)},...,W_K^{(j)}$, and the nodes in $\U = \{u_1,...,u_K\}$ behave as destinations, outputting $\hat W_1^{(j)},...,\hat W_K^{(j)}$ at the end of the block.
In the $(j+1)$th block, the nodes in $\U$ operate as sources for the second copy of $\N$, re-encoding the decoded messages from the previous block $\hat W_1^{(j)},...,\hat W_K^{(j)}$, and all the remaining nodes in the second copy of $\N$ simply operate according to coding scheme $\C_n$.
Provided that $\ep_n$ is small enough, $L^{-1} < \delta$, and at the end of the $[L(1-\delta) + 1]$th block, each destination $d_i$ obtains an estimate for all $L(1-\delta)$ messages from $s_i$.
By the union bound, the error probability of this code over $\N^2$ is at most $2 L(1-\delta) \ep_n \leq 2 \ep_n^{1/2}$, which tends to zero as $\ep_n \to 0$.
\end{proof}

\begin{lemma} \label{ranklem}
If $\bf x$ is a $d$-dimension random vector with entries in a finite field $\F$, then
\aln{
H(A {\bf x} | B{\bf x}) \leq \left( \rank \begin{bmatrix} A \\ B \end{bmatrix}  - \rank B \right) \log|\F|.
}
\end{lemma}

\begin{proof}
Let $B'$ be a $(\rank B) \times d$ matrix made up of $\rank B$ linearly independent rows of $B$.
Clearly, $H(B {\bf x}) = H(B'{\bf x})$.
Let $A'$ be a matrix obtained by removing rows of $A$ until $\rank \begin{bmatrix} A' \\ B' \end{bmatrix}$ is full rank.
We then have
%
\aln{
H(A{\bf x}|B{\bf x})&\leq H(A{\bf x}|B'{\bf x})  =H(A{\bf x},B'{\bf x})-H(B'{\bf x})\\
&=H\left( \begin{bmatrix} A\\B'\end{bmatrix}{\bf x}\right)-H(B'{\bf x}) \non
&=H\left( \begin{bmatrix} A'\\B'\end{bmatrix}{\bf x}\right)-H(B'x) \non
%
&=H(A'{\bf x},B'{\bf x})-H(B'x) \non
& \leq H(A'{\bf x}) \leq \rank(A') \log |\F|.
}
Moreover, we have that
\aln{
& \rank \begin{bmatrix} A \\ B \end{bmatrix}  - \rank B  = \rank \begin{bmatrix} A' \\ B' \end{bmatrix}  - \rank B' = \rank A',
}
concluding the proof.
\end{proof}

\vspace{2mm}

\begin{lemma} \label{truncranklem}
For a vector $\bf y$, let $\lfloor \bf y \rfloor$ be obtained by applying the floor function to each component of $\bf y$.
If $\bf x$ is a $d$-dimension zero-mean continuous random vector with $E[x_i^2]\leq P$, then
\aln{
H\left( \left\lfloor A {\bf x} \right \rfloor |\left\lfloor B{\bf x} \right \rfloor \right) \leq \left( \rank \begin{bmatrix} A \\ B \end{bmatrix}  - \rank B \right) \tfrac12 \log P + c,
}
where $c=o(\log P)$.
\end{lemma}

\begin{proof}
Following the proof of Lemma~\ref{ranklem}, we let $B'$ be a $(\rank B) \times d$ matrix made up of $\rank B$ linearly independent rows of $B$ and
$A'$ be a matrix obtained by removing rows of $A$ until $\rank \begin{bmatrix} A' \\ B' \end{bmatrix}$ is full rank.
Furthermore, we let $\tilde A$ be the matrix containing the $t$ rows removed from $A$ to obtain $A'$.
Notice that there exists a matrix $M$ such that $\tilde A = M \begin{bmatrix} A' \\ B' \end{bmatrix}$.
We then have
%
\al{
& \hspace{-1.6mm} H(\floor{A{\bf x}}| \floor{B{\bf x}})  \leq H(\floor{A{\bf x}}| \floor{B'{\bf x}}) \non
&=H\left( \floor{\begin{bmatrix} A\\B'\end{bmatrix}{\bf x}}\right)-H( \floor{B'{\bf x}}) \non
& =  I\left( {\bf x}; \floor{\begin{bmatrix} A\\B'\end{bmatrix}{\bf x}}\right)-I({\bf x}; \floor{B'{\bf x}}).  \label{truncranklem1}
}
Now if we let ${\bf z}_{A'}$, ${\bf z}_{B'}$ and ${\bf z}_{\tilde A}$ be independent random vectors of dimensions $\rank A'$, $\rank B'$ and $t$ respectively with i.i.d. $\N(0,1)$ entries.
Then, from Lemma~7.2 in \cite{ADT11}, we can upper-bound \eref{truncranklem1} by
\aln{
& I\left( {\bf x}; \begin{bmatrix} \tilde A \\ A' \\B'\end{bmatrix}{\bf x} + \begin{bmatrix} {\bf z}_{\tilde A} \\ {\bf z}_{A'} \\ {\bf z}_{B'}\end{bmatrix} \right)-I({\bf x}; B'{\bf x} + {\bf z}_{B'}) + c_1 \non
& = I\left( {\bf x}; A'{\bf x} + {\bf z}_{A'} |B'{\bf x} + {\bf z}_{B'}\right) \non
& \quad +I\left( {\bf x}; \tilde A {\bf x} + {\bf z}_{\tilde A} \left| \begin{bmatrix}  A' \\B'\end{bmatrix}{\bf x} + \begin{bmatrix}  {\bf z}_{A'} \\ {\bf z}_{B'}\end{bmatrix} \right. \right)
+ c_1 \non
& \leqnum I\left( {\bf x}; A'{\bf x} + {\bf z}_{A'} \right) \non
& \quad + I\left( {\bf x}; {\bf z}_{\tilde A} - M \begin{bmatrix}  {\bf z}_{A'} \\ {\bf z}_{B'}\end{bmatrix}  \left| \begin{bmatrix}  A' \\B'\end{bmatrix}{\bf x} + \begin{bmatrix}  {\bf z}_{A'} \\ {\bf z}_{B'}\end{bmatrix} \right. \right)
+ c_1 \non
& \leq I\left( {\bf x}; A'{\bf x} + {\bf z}_{A'} \right) \non
& \quad + h\left( {\bf z}_{\tilde A} - M \begin{bmatrix}  {\bf z}_{A'} \\ {\bf z}_{B'}\end{bmatrix} \right) - h\left( {\bf z}_{\tilde A} 
  \left| \begin{bmatrix}  A' \\B'\end{bmatrix}{\bf x} + \begin{bmatrix}  {\bf z}_{A'} \\ {\bf z}_{B'}\end{bmatrix}, {\bf x} \right. \right)
+ c_1 \non
& = I\left( {\bf x}; A'{\bf x} + {\bf z}_{A'} \right) + h\left( {\bf z}_{\tilde A} - M \begin{bmatrix}  {\bf z}_{A'} \\ {\bf z}_{B'}\end{bmatrix} \right) - h\left( {\bf z}_{\tilde A} \right)
+ c_1 \non
& \leq I\left( {\bf x}; A'{\bf x} + {\bf z}_{A'} \right) + c_1 + c_2,
} \rescnt
where \cnt follows from $\MC{A'{\bf x} + {\bf z}_{A'}}{\bf x}{B'{\bf x} + {\bf z}_{B'}}$ and $c_1$ and $c_2$ are scalars independent of $P$.
Since a MIMO channel with transfer matrix $A'$ has $\rank A'$ degrees of freedom, we have that \aln{ 
 I\left( {\bf x}; A'{\bf x} + {\bf z}_{A'} \right)  \leq  \left( \rank A' \right) \tfrac12 \log P + o(\log P).
}
Moreover, from the proof of Lemma \ref{ranklem}, we know that
$\rank A' = \rank \begin{bmatrix} A \\ B \end{bmatrix}  - \rank B$, which concludes the proof.
%
\end{proof}

\vspace{4mm}

\begin{lemma} \label{submatrixlem}
Let $A$ be an $n\times n$ invertible matrix.
If $A'$ is an $(n-1)\times(n-1)$ submatrix obtained by removing the $i$th row and $j$th column of $A$ for some $i$ and $j$, then $\rank A' \geq n-2$.
\end{lemma}

\vspace{2mm}

\begin{proof}
Suppose by contradiction that $\rank A' < n-2$.
Consider the cofactor expansion of the determinant of $A$ along the $i$th row.
For each element $(i,k)$, for $k \ne j$, the $(i,k)$th cofactor of $A$ corresponds to the determinant of a matrix $A''$, obtained by replacing one of the columns of $A'$ with the $j$th column of $A$ without the $i$th entry.
Since $\rank A' < n-2$, $\rank A'' \leq n-1$ and $\det A'' = 0$.
Moreover, the $(i,j)$th cofactor of $A$ is simply $\det A' = 0$.
But this implies that $\det A = 0$, which is a contradiction.
\end{proof}

\vspace{4mm}

\begin{lemmarep}{\ref{andlem}}
If a $\kkk$ AWGN network is diagonalizable, then for almost all values of the channel gains, $D_\Sigma = K$.
\end{lemmarep}

\vspace{2mm}

\begin{proof}
The achievability scheme used to achieve $K$ sum degrees of freedom is nearly identical to the Aligned Network Diagonalization scheme from \cite{dofkkk} in the case of constant channel gains.
We will point out the main differences and refer the reader to \cite{dofkkk} for the technical details.

Each source $s_i$ starts by breaking its message $W_{s_i}$ into $L$ submessages.
Each of the submessages will be encoded in a separate data stream, using a single codebook with codewords of length $n$  
and only integer symbols.
Now, let $\E_1$ and $\E_2$ be the edges from the first and second hops respectively.
Then we define $\Delta_N = \{0,...,N-1\}^{|\E_1|}$ and
\al{
T_{\vec m} = \prod_{(s_i,u_j) \in \E_1} F(s_i,u_j)^{m(s_i,u_j)},
}
for some $\vec m = ( m(e) : e \in \E_1 ) \in \IN^{|\E_1|}$, and the set of transmit directions for the first hop will be given by
%
\al{ \label{transdirections}
\T_N = \left\{  T_{\vec m} : \vec m \in \Delta_N \right\},
}
for some arbitrary $N$.
Notice that the number of transmit directions (which is also the number of data streams) is $L = |\T_N| = |\Delta_N| = N^{|\E_1|}$.
We will let $c_{i,\vec m}[1]$, $c_{i,\vec m}[2],...,c_{i,\vec m}[n]$ be the $n$ symbols of the codeword associated to the submessage to be sent by source $s_i$ over the transmit direction indexed by $\vec m$.
At time $t \in \{1,...,n\}$, source $s_i$ will thus transmit 
\aln{
X_{s_i}[t] = \ge \sum_{\vec m \in \Delta_N}  T_{\vec m} \; c_{i,\vec m} [t] 
} 
where $\ge$ is chosen to satisfy the power constraint.
%
%
%

The received signal at relay $u_j$ can be written as
\al{ \label{recnotaligned}
Y_{u_j}[t] & = \ge \sum_{\vec m \in \Delta_{N}}  T_{\vec m} \left( \sum_{i=1}^K F_{s_i,u_j} c_{i,\vec m}[t] \right) + Z_{u_j}[t] \non
& = \ge \sum_{\vec m \in \Delta_{N+1}}  T_{\vec m} \; q_{j,\vec m} [t] + Z_{u_j}[t],
}
where $q_{j,\vec m}[t] = \sum_{i=1}^K c_{i,\vec m_{ij}}[t]$ and we define 
$m_{ij}(s_k,u_\ell) = m(s_k,u_\ell)$ if $(s_k,u_\ell) \ne (s_i,u_j)$, $m_{ij}(s_i,u_j) = m(s_i,u_j) -1$ and
$c_{i,\vec m}[t] = 0$ if any component of $\vec m$ is $-1$ or $N$.
As explained in \cite{dofkkk}, for almost all values of the channel gains, relay $u_j$ can decode each integer $q_{j,\vec m}$ with high probability.
These integers will be re-encoded by $u_j$ using new transmit directions.
To describe the new set of transmit directions, we first define
\al{
B(\S,\U) = \begin{bmatrix} B(s_1,u_1) &  ... & B(s_K,u_1) \\ 
\vdots & \ddots & \vdots \\ 
B(s_1,u_K) & ... & B(s_K,u_K) \end{bmatrix} = F(\U,\D)^{-1}.
}
Since we are considering a diagonalizable $\kkk$ network according to Definition \ref{diagdef}, for almost all values of the channel gains, $B(s_i,u_j) \ne 0$ if and only if $(s_i,u_j) \in \E_1$.
%
Thus, we may let
\al{ \label{directionsfixed2}
\til T_{\vec m} = \prod_{ (s_i,u_j) \in \E_1 } B(s_i,u_j)^{m(s_i,u_j)},
}
and, similar to (\ref{transdirections}), we can define the set of transmit directions for the relays to be
\aln{
\til \T_{N+1} = \left\{  \til T_{\vec m} : \vec m \in \Delta_{N+1} \right\}.
}
Relay $u_j$ will re-encode the $q_{j,\vec m}$\,s by essentially replacing each received direction $T_{\vec m}$ in (\ref{recnotaligned}) with the direction $\til T_{\vec m}$.
We highlight that this is only possible under the assumption of a diagonalizable $\kkk$ network.
The transmit signal of relay $u_j$ at time $t+1$ will be given by
\al{ \label{transaligned} 
X_{u_j}[t+1] & = \ge' \sum_{\vec m \in \Delta_{N+1}}  \til T_{\vec m} \; q_{j,\vec m} [t] \non
& = \ge' \sum_{\vec m \in \Delta_{N}}  \til T_{\vec m} \left( \sum_{i=1}^K B(s_i,u_j) \, c_{i,\vec m} [t] \right)
}
where $\ge'$ is chosen so that the output power constraint is satisfied.
%
%
%
%
%
%

In order to compute the received signals at the destinations, we first notice that, from (\ref{transaligned}), the vector of the $K$ relay transmit signals at time $t+1$ can be written as
\al{ 
& \ge' \sum_{\vec m \in \Delta_{N}}  \til T_{\vec m} 
\begin{bmatrix} B(s_1,u_1) &  ... & B(s_K,u_1) \\ 
\vdots & \ddots & \vdots \\ 
B(s_1,u_K) & ... & B(s_K,u_K) \end{bmatrix}
\begin{bmatrix} c_{1,\vec m}[t] \\ \vdots \\ c_{K,\vec m}[t] \end{bmatrix}. \label{relaytransmitvector}
}

Since the $\tilde T_{\vec s}$\,s are just scalars, we can write the vector of the $K$ received signals at the destinations as
\aln{
& \begin{bmatrix} Y_{d_1}[t+1]\\ \vdots \\ Y_{d_K}[t+1] \end{bmatrix} 
= F(\U,\D)
\begin{bmatrix} X_{u_1}[t+1] \\ \vdots \\ X_{u_K}[t+1] \end{bmatrix} + \begin{bmatrix} Z_{d_1}[t+1] \\ \vdots \\ Z_{d_K}[t+1] \end{bmatrix} }
\aln{
& = B(\S,\U)^{-1}
\begin{bmatrix} X_{u_1}[t+1]  \\ \vdots \\ X_{u_K}[t+1] \end{bmatrix} + \begin{bmatrix} Z_{d_1}[t+1] \\ \vdots \\ Z_{d_K}[t+1] \end{bmatrix} \\
&= \ge' \sum_{\vec s \in \Delta_{N}}  \til T_{\vec s} 
\begin{bmatrix} c_{1,\vec m}[t] \\ \vdots \\ c_{K,\vec m}[t] \end{bmatrix}+ \begin{bmatrix} Z_{d_1}[t+1] \\ \vdots \\ Z_{d_K}[t+1] \end{bmatrix}.
}
Thus, the received signal at destination $d_j$ at time $t+1$ is simply given by
\al{ \label{destreceivedscalar}
Y_{d_j}[t+1] = \ge' \sum_{\vec m \in \Delta_{N}}  \til T_{\vec m} \; c_{j,\vec m}[t] + Z_{d_j}[t+1],
}
and we see that all the interference has been cancelled, and destination $d_j$ receives only the data streams originated at source $s_j$.
Following the arguments in \cite{dofkkk}, it can be shown that such a scheme can indeed achieve $K$ DoF.
\end{proof}

\fi

\section*{Acknowledgements}

The research of A. S. Avestimehr and Ilan Shomorony is supported by a 2013 Qualcomm Innovation Fellowship, NSF Grants CAREER 1408639, CCF-1408755, NETS-1419632, EARS-1411244, ONR award N000141310094.

{\footnotesize
\bibliographystyle{IEEEtran}
\bibliography{refsl}
}

\end{document}